\documentclass[runningheads]{llncs}
\usepackage{amssymb}
\usepackage{textcomp}
\usepackage{amsmath}
\usepackage{mathtools}

\usepackage{amsthm}
\usepackage{algorithm}
\usepackage[noend]{algpseudocode}
\usepackage{comment}
\usepackage{graphicx}
\usepackage{color}
\usepackage{tikz}
\usepackage{cleveref}
\DeclareMathOperator{\argmin}{arg\,min}
\DeclareMathOperator{\dist}{dist}
\DeclareMathOperator{\LCA}{LCA}
\newcommand{\ket}[1]{\left| #1 \right\rangle}

\begin{document}

\title{Fast Classical and Quantum Algorithms for Online $k$-server Problem on Trees
\thanks{
This work on Quantum algorithm for k-server problem was supported by Russian Science Foundation Grant 19-19-00656.
The research in Section 4 is funded by the subsidy allocated to Kazan Federal University for the state assignment in the sphere of scientific activities, project No. 0671-2020-0065. The research is also supported in part by the ERA-NET Cofund in Quantum Technologies project QuantAlgo and the French ANR Blanc project RDAM.}}
%
%
\author{Ruslan Kapralov\inst{1} \and
Kamil Khadiev\inst{1,2} \and
Joshua Mokut\inst{1} \and
Yixin Shen\inst{3} \and
Maxim Yagafarov\inst{1}}
\authorrunning{R. Kapralov, et al.}
%
\institute{Kazan Federal University, Kazan, Russia\and Zavoisky Physical-Technical Institute, FRC Kazan Scientific Center of RAS, Kazan, Russia \and
Universit\'{e} de Paris, CNRS, IRIF, F-75006 Paris, France}
\maketitle              
\begin{abstract}
We consider online algorithms for the $k$-server problem on trees. 
Chrobak and Larmore proposed a $k$-competitive algorithm for this problem that has the optimal competitive ratio. However, a naive implementation of their algorithm has time complexity $O(n)$
to process each request, where $n$ is the number of nodes. We propose a new time-efficient implementation of this algorithm that has $O(n\log n)$ time complexity for preprocessing and $O\left(k^2 + k\cdot \log n\right)$ time for processing a request. We also propose a quantum algorithm for the case where the nodes of the tree are presented using string paths. In this case, no preprocessing is needed, and the time complexity for each request is $O(k^2\sqrt{n}\log n)$. When the number of requests is $o\left(\frac{\sqrt{n}}{k^2}\right)$, we obtain a quantum speed-up on the total runtime compared to our classical algorithm.

We also give a simple quantum algorithm to find the first marked element in a collection
of $m$ objects, that works even in the presence of two-sided bounded errors on the input oracle. It has worst-case query complexity $O(\sqrt{m})$. In the particular case of one-sided errors on the input, it has expected query complexity $O(\sqrt{x})$ where $x$ is the position of the first marked element. Compared with previous work, our algorithm can handle errors in the input oracle.

\end{abstract}
\section{Introduction}
\label{sec:intro}
Online optimization is a field of optimization theory that deals with optimization problems having no knowledge of the future \cite{k2016}. An online algorithm reads an input piece by piece and returns an answer piece by piece immediately, even if the answer can depend on future pieces of the input. The goal is to return an answer that minimizes an objective function (the cost of the output). The most standard method to define the effectiveness of an online algorithm is the competitive ratio \cite{kmrs86}.
The competitive ratio is the approximation ratio achieved by the algorithm. That is the worst-case ratio between the cost of the solution found by the algorithm and the cost of an optimal solution. If the ratio is $c$, then the online algorithm is called $c$-competitive. %
In the general setting, online algorithms have unlimited computational power. Nevertheless, many papers consider them with different restrictions. Some of them are restrictions on memory 
\cite{bk2009,blm2015,kkm2018,aakv2018,gs93,h95,kk2019disj,k2021,kl2020,kk2019online2w,kk2019},
others are restrictions on time complexity \cite{fnn2006,rbm2013}.

In this paper, we focus on efficient online algorithms in terms of time complexity.
We consider the $k$-server problem on trees.
 Chrobak and Larmore \cite{cl91} proposed a $k$-competitive algorithm for this problem where the competitive ratio $k$ is the best possible for deterministic algorithms for this problem. The existing implementation of their algorithm has $O(n)$ time complexity for each request, where $n$ is the number of nodes in the tree. For general graphs, there exists a time-efficient algorithm for the $k$-server problem \cite{rbm2013} that uses min-cost-max-flow algorithms. However, in the special case of a tree, this algorithm is not optimal. We propose a new time-efficient implementation of the algorithm from \cite{cl91}. It has $O\left(n\log n\right)$ time complexity for preprocessing and $O\left(k^2+k\log n \right)$ for processing a request. It is based on fast algorithms for computing Lowest Common Ancestor (LCA) \cite{bfc2000} and the binary lifting technique \cite{bf2004}. 


We also revisit the problem of finding the first marked element in
a collection of $m$ objects. It is well-known that it can be solved
in expected $O(\sqrt{m})$ queries when given quantum oracle access to the input,
and expected $O(\sqrt{x})$ queires where $x$ is the position
of the first marked element \cite{dh96},\cite[Theorem~10]{ll2015} and \cite[Theorem~6]{k2014}.
We give a $O(\sqrt{m})$ queries algorithm that works
even in the presence of \emph{two-sided bounded errors} in the input. We also
provide an expected $O(\sqrt{x})$ queries algorithm in the case where the
input has \emph{one-sided errors} only. Note that contrary to the previous results,
we assume that the oracle can have errors whereas \cite{dh96,ll2015,k2014} assume
that the oracle is perfect.

We then consider the $k$-server problem in the case where the description of the tree is given by string paths. The string path of a node in a rooted tree is a sequence of length $h$, where $h$ is the height of the node, describing the path from the root to the node. Such a way of representing the trees is useful, for example, as a path to a file in file systems.
Assuming oracle access to the length of the string path of each node and to each element on the string path with time complexity $O(1)$, we obtain a quantum algorithm for this problem with $O(k^2\sqrt{n}\log(n))$ running time to process each request and without preprocessing. This algorithm is based on our improved quantum search algorithm. When the number of requests is $o\left(\frac{\sqrt{n}}{k^2}\right)$,
the total runtime of our quantum algorithm is smaller than the classical one. 

The structure of the paper is the following. Section \ref{sec:prelims} contains preliminaries. The classical algorithm is described in Section \ref{sec:c-algorithm}. Section \ref{sec:binsearch} contains our improved 
quantum search algorithm. The quantum algorithm for the $k$-server problem is described in Section \ref{sec:q-algorithm}. 

\section{Preliminaries}
\label{sec:prelims}
\subsection{Online algorithms}
{\bf An online minimization problem} consists of a set $\cal{I}$ of inputs and a cost function. Each input $I = (x_1, \dots , x_n)$ is a sequence of requests, where $n$ is the length of the input $|I|=n$. Furthermore, a set of feasible outputs (or solutions) $\widetilde{O}(I)$ is associated with each $I$; an output is a sequence of answers $O = (y_1, \dots, y_n)$. The cost function assigns a positive real value $cost(I, O)$ to $I\in{ \cal I}$ and $O\in\widetilde{O}(I)$. An optimal  solution for $I\in{\cal I}$ is $O_{opt}(I)=\argmin_{O\in\widetilde{O}(I)}cost(I,O)$.

Let us define an online algorithm for this problem.
{\bf A deterministic online algorithm}  $A$ computes the output sequence $A(I) = (y_1,\dots , y_n)$ such that $y_i$ is computed based on $x_1, \dots , x_i$.  
We say that $A$ is $c$-{\em competitive} if there exists a constant $\alpha\geq 0$ such that, for every $n$ and for any input $I$ of size $n$, we have: $cost(I,A(I)) \leq c \cdot cost(I,O_{Opt}(I)) + \alpha$. The minimal $c$ that satisfies the previous condition is called the {\bf competitive ratio} of $A$. 

\subsection{Rooted Trees}\label{sec:lca}\label{sec:bl}
Consider a rooted tree $G=(V,E)$, where $V$ is the set of nodes/vertices, and $E$ is the set of edges. Let $n=|V|$ be the number of nodes, or equivalently the size of the tree. We denote by $1$ the root of the tree.
A path $P$ is a sequence of nodes $(v_1,\dots,v_h)$ that are connected by edges, i.e. $(v_i,v_{i+1})\in E$ for all $i\in\{1,\dots,h-1\}$, such that there are no duplicates among $v_1,\dots,v_h$. Here $h$ is a length of the path. Between any two nodes $v$ and $u$ on the tree, there is a unique path. The distance $\dist(v,u)$ is the length of this path.
For each node $v$ we can define a parent node $\textsc{Parent}(v)$ which is the first node on the unique path from $v$ to root $1$.  We have $\dist\left(1,\textsc{Parent}(v)\right)+1=\dist(1,v)$. Additionally, we can define the set of children $\textsc{Children}(v)=\{u: \textsc{Parent}(u)=v\}$. Any node $y$ on the unique path from root $1$ to node $v$ is an ancestor of node $v$. 


{\bf Lowest Common Ancestor (LCA).}
Given two nodes $u$ and $v$ of a rooted tree, the Lowest Common Ancestor is the node $w$ such that $w$ is an ancestor of both $u$ and $v$, and $w$ is the closest one to $u$ and $v$ among all such ancestors. The following result is well-known.

\begin{lemma}[\cite{bfc2000}]\label{lm:lca}
There is an algorithm for the LCA problem with the following properties:
    (i) The time complexity of the preprocessing step is $O(n)$;
    (ii)The time complexity of computing LCA for two vertices is $O(1)$.
\end{lemma}

We call $\textsc{LCA\_Preprocessing}()$ the subroutine that does the preprocessing for the algorithm and $\textsc{LCA}(u,v)$ that computes the LCA of two nodes $u$ and $v$.


{\bf Binary Lifting Technique.}
This technique from \cite{bf2004} allows us to obtain a vertex $v'$ that is at distance $z$ from a vertex $v$ with $O(\log n)$ time complexity. There are two procedures:
    
    $\bullet\quad \textsc{BL\_Preprocessing}()$ prepares the required data structures. The time complexity is $O(n\log n)$.
    
    $\bullet\quad\textsc{MoveUp}(v,z)$ returns a vertex $v'$ on the path from $v$ to the root and at distance $\dist(v',v)=z$. 
    The time complexity is $O(\log n)$.

 The technique is well documented in the literature. We present an implementation in Appendix \ref{apx:bl} for completeness.

\subsection{$k$-server Problem on Trees}
Let $G=(V,E)$ be a rooted tree, and we are given $k$ servers that can move among nodes of $G$. At each time slot, a request $q\in V$
appears. We have to ``serve'' this request, that is, to choose one of the $k$ servers and
move it to $q$. The other servers are also allowed to move. The cost function is the distance by which we move the servers.
In other words, if before the request, the servers are at positions $v_1,\dots,v_k$ and after the request they are at $v'_1,\dots,v'_k$, then $q\in\{v'_1,\dots,v'_k\}$ and the cost of the move is $\sum_{i=1}^k \dist(v_i,v'_i)$. The cost of a sequence of requests is the sum of the costs of serving each requests. The problem is to design a strategy that minimizes the cost
of servicing a sequence of requests given online.

\subsection{Quantum query model}

We use the standard form of the quantum query model. 
Let $f:D\rightarrow \{0,1\},D\subseteq \{0,1\}^m$ be an $m$ variable function. We wish to compute on an input $x\in D$. We are given an oracle access to the input $x$, i.e. it is realized by a specific unitary transformation usually defined as $\ket{i}\ket{z}\ket{w}\rightarrow \ket{i}\ket{z+x_i\pmod{2}}\ket{w}$ where the $\ket{i}$ register indicates the index of the variable we are querying, $\ket{z}$ is the output register, and $\ket{w}$ is some auxiliary work-space. An algorithm in the query model consists of alternating applications of arbitrary unitaries independent of the input and the query unitary, and a measurement in the end. The smallest number of queries for an algorithm that outputs $f(x)$ with probability $\geq \frac{2}{3}$ on all $x$ is called the quantum query complexity of the function $f$ and is denoted by $Q(f)$. We refer the readers to \cite{nc2010,a2017,aazksw2019part1} for more details on quantum computing. 



\begin{definition}[Search problem]
Suppose we have a set of objects named $\{1,2,\dots, m\}$, of which some are targets. Suppose  $\mathcal{O}$ is an oracle that identifies the targets.
The goal of a search problem is to find a target $i \in \{1,2,\dots, m\}$ by making queries to the oracle $\mathcal{O}$.
\end{definition}

In search problems, one will try to minimize the number of queries to the oracle. In the classical setting, one needs $O(m)$ queries to solve such a problem. Grover, on the other hand, constructed a quantum algorithm that solves the search problem with only $O(\sqrt{m})$ queries~\cite{g96}, provided that there is a unique target.  When the number of targets is unknown, Brassard \emph{et al.} designed a modified Grover algorithm that solves the search problem with $O(\sqrt{m/\lambda})$ queries~\cite{bbht98}, where $\lambda$ is the number of targets, which is of the same order as the query complexity of the Grover search.

\section{A Fast Online Algorithm for $k$-server Problem on Trees with Preprocessing}\label{sec:c-algorithm}

We first describe Chrobak-Larmore's $k$-competitive algorithm for $k$-server problem on trees from \cite{cl91}.
Assume that we have a request on a vertex $q$, and the servers are on the vertices $v_1,\dots,v_k$.
We say that a server $i$ is {\em active} if there are no other servers on the path from $v_i$ to $q$.
 In each phase, we move every {\em active} server one step towards the vertex $q$. After each phase, the set of {\em active} servers can change.
 We repeat this phase (moving of the active servers) until one of the servers reaches the queried vertex $q$.

The naive implementation of this algorithm has time complexity $O(n)$ for each request. First, we run a depth-first search with time labels \cite{cormen2001}, whose result allows us to check in constant time whether a vertex $u$ is an ancestor of a vertex $v$.
Recall that time labels record the first timestamp $f(v)$
when the depth-first search enters a node $v$, and the last timestamp $\ell(v)$
when the depth-first search finishes processing the last child of
node $v$; the timestamp increases every time a new node is visited.
A node $u$ is then an ancestor of $v$ if and only if the interval
$[f(v),\ell(v)]$ is contained in $[f(u),\ell(u)]$.
After that, we can move each active server towards the queried vertex, step by step. Together all active servers cannot visit more than $O(n)$ vertices.

In the following, we present an effective implementation of Chrobak-Larmore's algorithm with preprocessing. The preprocessing part is done once and has $O(n\log n)$ time complexity (Theorem \ref{th:preproc}). The request processing part is done for each request and has $O\left(k^2 + k\cdot\log n\right)$ time complexity  (Theorem \ref{th:query-proc}).
 
\subsection{Preprocessing}
We do the following steps for the preprocessing 
:
    
    1. We do required preprocessing for LCA algorithm (Section \ref{sec:lca})
    
     2. We do required preprocessing for Binary lifting technique (Section \ref{sec:bl}) 
    
     3. Additionally, for each vertex $v$ we compute the distance from the root to $v$, i.e. $\dist(1,v)$. This can be done using a depth-first search algorithm \cite{cormen2001}. We store all the distances in an array. See Appendix \ref{apx:alg:dist} for an implementation of $\textsc{ComputeDistance}(u)$.



\begin{algorithm}[ht]
    \caption{$\textsc{Preprocessing}$. The preprocessing procedure.} \label{alg:preproc}
    \begin{algorithmic}
        \State $\textsc{LCA\_Preprocessing()}$
        \State $\textsc{BL\_Preprocessing()}$
        \State $\dist(1,1)\gets 0$
        \State $\textsc{ComputeDistance(1)}$
    \end{algorithmic}
\end{algorithm}

\begin{theorem}\label{th:preproc}
The preprocessing has time complexity $O(n\log n)$.
\end{theorem}
\begin{proof}
The time complexity of the preprocessing phase is $O(n)$ for LCA, $O(n\log n)$ for the binary lifting technique and $O(n)$ for $\textsc{ComputeDistance}(1)$. Therefore, the total time complexity is $O(n\log n)$.
\end{proof}
\subsection{Request Processing}
Assume that we have a request on a vertex $q$, and the servers are on the vertices $v_1,\dots,v_k$. We do the following steps, implemented in Algorithms~\ref{alg:query}~and~\ref{alg:move}.

\noindent
 {\bf Step 1.} We sort all the servers by their distance to the node $q$. The distance $\dist(v,q)$ between a node $v$ and the node $q$ can be computed in the following way. Let $l=\textsc{LCA}(v,q)$ be the lowest common ancestor of $v$ and $q$, then $\dist(v,q)=\dist(1,q)+\dist(1,v)-2\cdot \dist(1,l)$. Using the prepocessing, this quantity can be computed in constant time.
 We denote by $\textsc{Sort}(q,v_1,\dots,v_k)$ this sorting procedure.
     In the following steps we assume that $\dist(v_i,q)\leq \dist(v_{i+1},q)$ for $i\in\{1,\dots,k-1\}$.
    
    \noindent
     {\bf Step 2.}  The first server on $v_1$ processes the request. We move it to the node $q$.
     
     \noindent
     {\bf Step 3.} For $i\in\{2,\dots k\}$ we consider the server on $v_i$. It will be inactive when some other server with a smaller index arrives on the path between $v_i$ and $q$. Section \ref{sec:distance} contains the different cases that can happen and how to compute the distance $d$
     traveled by $v_i$ before it becomes inactive. We then move the $i$-th
     server $d$ steps towards the request $q$. The new position of the $i$-th server is $v'_i$.
     
    \begin{algorithm}[ht]
    \caption{$\textsc{Request}(q)$. Request processing procedure.} \label{alg:query}
    \begin{algorithmic}
        \State $\textsc{Sort}(q,v_1,\dots,v_k)$ 
        \State $v'_1\gets q$
        \For{$i\in\{2,\dots,k\}$}
            \State $d\gets \textsc{DistanceToInactive}(q,i)$ \Comment{see Algorithm~\ref{alg:dist2inactive}}
            \State $v'_i \gets \textsc{Move}(v_i,d)$ \Comment{see Algorithm~\ref{alg:move}}
        \EndFor
    \end{algorithmic}
\end{algorithm}

\vspace{-1cm}
\subsection{Distance to Inactive State}\label{sec:distance}

When processing a request, all servers except one will eventually become inactive. The crucial part of the optimization is to compute when a server becomes inactive quickly. For the purpose of computing this time,
we claim that we can pretend that servers ``never go inactive''.
Formally, let $q$ be a request, $i$ be a server, and $j$ another server with smaller index.
We know that $i$ will become inactive because it is not the closest to the
target. However it is possible that this particular server $j$ is not the
one that will render $i$ inactive. Nevertheless, we can pretend that $j$
will never become inactive and compute the distance $i$ will travel before
going inactive because of $j$, call this distance $d_{i,j}^q$ (the index $i$
is fixed in this reasoning). We claim the following:

\begin{lemma}\label{lm:inactive}
    For any request $q$ and server $i>1$ (\emph{i.e.} a server that will become
    inactive), the distance $D_i^q$ travelled by $i$ before it becomes inactive
    is equal to $\min_{j<i}d_{i,j}^q$. 
\end{lemma}

\begin{proof}
    Let $j_0$ be one of the servers that renders $i$ inactive, then
    $D_i^q=d_{i,j_0}^q$ because $j_0$ will not become inactive before
    it makes $i$ inactive, hence for the purpose of computing $D_i$, it makes
    no difference whether $j_0$ eventually becomes inactive or not.
    Therefore, we only need to prove no other $d_{i,j}^q$ is strictly smaller.
    Assume for contradiction that $d_{i,j}^q<D_i^q$ for some $j<i$, and pick $j$ so that $d_{i,j}^q$ is minimum among all $j<i$ (and in case of equality,
    pick $j$ the smallest possible). Then, it means there exists a vertex $t$
    such that $d_{i,j}^q=\dist(v_j,t)\leqslant\dist(v_i,t)$
    and $t$ is on the paths from $v_i$ and $v_j$ to $q$.
    Now we claim that $j$ must become inactive
    before it reaches $t$. Indeed, if not, it would reach $t$ and makes
    $i$ inactive after a distance $d_{i,j}^q<D_i^q$, which is impossible by
    definition of $D_i^q$. Therefore, $j$ is rendered inactive
    before reaching $t$ by another server $\ell$ reaching some vertex $u$ on
    the path from $v_j$ to $t$.
    In particular, we must have $D_j^q=\dist(v_\ell,u)\leqslant\dist(v_j,u)$ and
    $\dist(v_j,u)<\dist(v_j,t)$.
    But now observe that if we pretend that $\ell$ never goes inactive,
    it will reach $t$ after travelling a distance
    $\dist(v_\ell,u)+\dist(u,t)\leqslant\dist(v_j,u)+\dist(u,t)=\dist(v_j,t)$
    hence
    $d_{i,\ell}^q=\dist(v_\ell,t)\leqslant\dist(v_j,t)=d_{i,j}^q$.
    But we chose $j$ so that $d_{i,j}^q$ is minimal so we must have
    $d_{i,\ell}^q=d_{i,j}^q$ and therefore, $j<\ell$ (we sort by index
    in case of tie). Going back to the computation, we see that
    $d_{i,\ell}^q=d_{i,j}^q$ implies that $\dist(v_i,u)=\dist(v_\ell,u)$,
    \emph{i.e.} $i$ and $\ell$ reach $u$ at the same time. But when two
    servers reach the same vertex simultaneously, the greater index goes
    inactive, \emph{i.e.} $\ell$ would go inactive because of $j$.
    This is a contradiction because we assumed that $\ell$ is the one
    making $j$ inactive.
\end{proof}

We have now reduced the problem to the following question:
given a server $i$ and another server $j$ with smaller index, compute
$d_{i,j}^q$, the distance until $i$ becomes inactive because of $j$,
pretending that $j$ never goes inactive.
There are several cases to consider, depicted in
Figure~\ref{fig:distance_inactive}, depending on the relationship of $v_i$, $v_j$ and $q$
in the tree. Let $t$ be the vertex where the paths from $v_i$ to $q$
and $v_j$ to $q$ intersect the first time, then $d^q_{i,j}=\dist(v_j,t)$
and
\begin{enumerate}
    \item\label{enum:dist:case_1}
        if $q$ is an ancestor of $v_i$ and $v_j$, then $t=\LCA(v_i,v_j)$;
    \item\label{enum:dist:case_2}
        if $q$ is an ancestor of $v_i$ but not of $v_j$, then $t=q$;
    \item\label{enum:dist:case_3}
        if $v_i$ is an ancestor of $q$, then $t=\LCA(q,v_j)$ because
        $v_i$ must also be ancestor of $v_j$ since $v_j$ is closer to $q$
        than $v_i$;
    \item\label{enum:dist:case_4}
        if the LCA of $v_j$ and $q$ is not an ancestor of $v_i$,
            then $t=\LCA(v_j,q)$;
    \item\label{enum:dist:case_5}
        if the LCA of $v_i$ and $v_j$ is not an ancestor of $q$, then
        $t=\LCA(v_j,v_i)$;
    \item\label{enum:dist:case_6}
        otherwise $t=\LCA(v_i,q)$.
\end{enumerate}

Note that in this case distinction, the order of the cases is important: if cases 1 to 3 do not apply for example, then we know that $v_i$ is not an
ancestor of $q$ and $q$ is not an ancestor of $v_i$.

\begin{algorithm}[ht]
    \caption{$\textsc{DistanceToInactive}(q,i)$. Compute the distance travelled before going inactive.} \label{alg:dist2inactive}
    \begin{algorithmic}
        \State $d\gets\infty$
        \For{$j\in\{1,\dots,i-1\}$}
            \State $t\gets $ do case analysis as above
            \State $d\gets\min(d,\dist(t,v_j))$
        \EndFor
        \State\Return{$d$}
    \end{algorithmic}
\end{algorithm}

\vspace{-0.3cm}
\begin{lemma}\label{lm:dist2inactive}
    $\textsc{DistanceToInactive}$ has time complexity
    $O\left(k\right)$. 
\end{lemma}
\begin{proof}
    Since a vertex $u$ is ancestor of $v$ if $\LCA(u,v)=u$, we can check this
    condition in $O(1)$ due to results from Section \ref{sec:lca}.
    It follows that we can compute $d_{i,j}^q$ for every $i,j,q$ in $O(1)$
    and there are at most $k$ other servers to consider.
\end{proof}

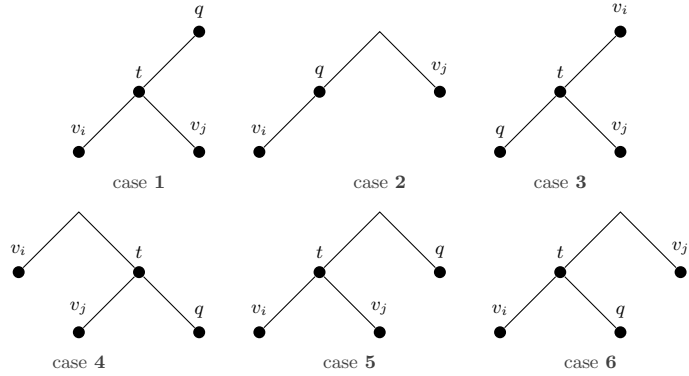
\begin{figure}
    \centering
    \begin{tikzpicture}[scale=0.8,transform shape,myvertex/.style={circle,fill=black,inner sep=2pt},
            mycase/.style={color=darkgray}]
        \begin{scope}[shift={(1,0)}]
            \draw[myvertex] node[myvertex,label=$q$] (q) {}
                node[myvertex,label=$t$] (t) at (-1,-1) {}
                node[myvertex,label=$v_i$] (vi) at (-2,-2) {}
                node[myvertex,label=$v_j$] (vj) at (0,-2) {};
            \draw (vi) -- (t) -- (q);
            \draw (vj) -- (t);
            \node[mycase] at (-1,-2.5) {case~\textbf{\ref{enum:dist:case_1}}};
        \end{scope}
        \begin{scope}[shift={(4,0)}]
            \draw[myvertex] coordinate (t) {}
                node[myvertex,label=$q$] (q) at (-1,-1) {}
                node[myvertex,label=$v_i$] (vi) at (-2,-2) {}
                node[myvertex,label=$v_j$] (vj) at (1,-1) {};
            \draw (vi) -- (q) -- (t) -- (vj);
            \node[mycase] at (0,-2.5) {case~\textbf{\ref{enum:dist:case_2}}};
        \end{scope}
        \begin{scope}[shift={(8,0)}]
            \draw[myvertex] node[myvertex,label=$v_i$] (vi) {}
                node[myvertex,label=$t$] (t) at (-1,-1) {}
                node[myvertex,label=$q$] (q) at (-2,-2) {}
                node[myvertex,label=$v_j$] (vj) at (0,-2) {};
            \draw (vi) -- (t) -- (q);
            \draw (vj) -- (t);
            \node[mycase] at (-1,-2.5) {case~\textbf{\ref{enum:dist:case_3}}};
        \end{scope}
        \begin{scope}[shift={(-1,-3)}]
            \draw[myvertex] coordinate (up) {}
                node[myvertex,label=$v_i$] (vi) at (-1,-1) {}
                node[myvertex,label=$t$] (t) at (1,-1) {}
                node[myvertex,label=$q$] (q) at (2,-2) {}
                node[myvertex,label=$v_j$] (vj) at (0,-2) {};
            \draw (vi) -- (up) -- (t) -- (q);
            \draw (t) -- (vj);
            \node[mycase] at (0,-2.5) {case~\textbf{\ref{enum:dist:case_4}}};
        \end{scope}
        \begin{scope}[shift={(4,-3)}]
            \draw[myvertex] coordinate (up) {}
                node[myvertex,label=$v_i$] (vi) at (-2,-2) {}
                node[myvertex,label=$t$] (t) at (-1,-1) {}
                node[myvertex,label=$q$] (q) at (1,-1) {}
                node[myvertex,label=$v_j$] (vj) at (0,-2) {};
            \draw (vi) -- (t) -- (up) -- (q);
            \draw (t) -- (vj);
            \node[mycase] at (-0.5,-2.5) {case~\textbf{\ref{enum:dist:case_5}}};
        \end{scope}
        \begin{scope}[shift={(8,-3)}]
            \draw[myvertex] coordinate (up) {}
                node[myvertex,label=$v_i$] (vi) at (-2,-2) {}
                node[myvertex,label=$t$] (t) at (-1,-1) {}
                node[myvertex,label=$v_j$] (vj) at (1,-1) {}
                node[myvertex,label=$q$] (q) at (0,-2) {};
            \draw (vi) -- (t) -- (q);
            \draw (t) -- (up) -- (vj);
            \node[mycase] at (-0.5,-2.5) {case~\textbf{\ref{enum:dist:case_6}}};
        \end{scope}
    \end{tikzpicture}
    \caption{List of the various cases to consider when computing the distance before a server $i$ is rendered inactive by a (closer to the request) server $j$.\label{fig:distance_inactive}}
\end{figure}

\subsection{How to Move a Server}\label{sec:move}

We now consider the following problem: given a server $v$ and a distance $z$, how to efficiently compute the new position of the server after moving it $z$ steps towards $q$.  
We use the binary lifting technique for this procedure.

Let $l=LCA(v,q)$. If $\dist(l,v)\geq z$, then the result node is on the path between $v$ and $l$. We can thus invoke $\textsc{MoveUp}(v,z)$ from Section \ref{sec:bl}. Otherwise, we should move the server first to $l$. We then move it $z-\dist(l,v)$ steps 
down towards $q$ from $l$. Moving down from $l$ is the same as moving up $\dist(l,q)-(z-\dist(l,v))$ steps from $q$. The algorithm is presented in Algorithm \ref{alg:move}.
  \begin{algorithm}[ht]
    \caption{$\textsc{Move}(v,z)$. Moves of a server from $v$ to distance $z$ on a path from $v$ to $q$.} \label{alg:move}
    \begin{algorithmic}
        \State $l=\textsc{LCA}(v,q)$
        \If{ $z\leq \dist(l,v)$}
        \State $Result \gets \textsc{MoveUp}(v,z)$
        \EndIf
       \If{ $z>\dist(l,v)$}
        \State $z\gets z-\dist(l,v)$
        \State $Result \gets \textsc{MoveUp}(q,\dist(l,q)-z)$
        \EndIf
        \State \Return{$Result$}
    \end{algorithmic}
\end{algorithm}

\begin{lemma}\label{lm:move}
    The time complexity of \textsc{Move} is $O\left(\log n\right)$.
\end{lemma}
\begin{proof}
The time complexity of $\textsc{MoveUp}$ is $O(\log n)$ using the binary lifting technique from \Cref{sec:bl} and $\textsc{LCA}$ is in $O(1)$
by \Cref{sec:lca}. Furthermore, we can compute the distance between any two nodes in $O(1)$ thanks to the preprocessing.
Therefore, the total complexity is $O(\log n)$.
\end{proof}

\begin{theorem}\label{th:query-proc}
The time complexity of the request processing phase is $O\left(k^2 + k\log n\right)$.
\end{theorem}
\begin{proof}
The complexity of sorting the servers by distance is $O(k\log k)$.
For each server, we compute the distance traveled before being inactive
in $O(1)$ by \Cref{lm:dist2inactive}. We then move each server by that
distance in time $O(\log n)$ by \Cref{lm:move}.
Therefore, the complexity of processing one server is $O(k+\log n)$,
and there are $k$ servers.
\end{proof}

\section{Binary Search for a Function with Errors }\label{sec:binsearch}

Consider a search space $S=\{1,\dots,m\}$ and a subset $M\subseteq S$
of marked elements. Define the indicator function $g_M:S^2\to\{0,1\}$ by
    $g_M(\ell,r)=1\text{ if }\{\ell,\ldots,r\}\cap M\neq\varnothing,
        \text{ and 0 otherwise}.$
In other words, $g_M(\ell,r)$ indicates whether there is a marked element
from $M$ in the interval $[\ell,r]$.
Now assume that we do not know $M$ but have access to a two-sided
probabilistic approximation $\tilde{g}$ of $g_M$.
Formally, there is a probability $p<1/2$ such that for any
$\ell,r\in S$,
$\tilde{g}(\ell,r)= g_M(\ell,r)$ with probability at least $1-p$ and $1-g_M(\ell,r)$ otherwise.
%
Intuitively, $\tilde{g}$ behaves like $g_M$ with probability at least $1-p$.
However, sometimes it makes mistakes and returns a completely wrong answer.
Note that $\tilde{g}$ has \emph{two-sided} error:
it can return $0$ even if the interval $[\ell,r]$ contains a marked element, but more importantly, it can also return $1$ even though the interval
does \emph{not} contain any marked element.
We further assume that a call to $\tilde{g}(\ell,r)$ takes time
$T(r-\ell)$ where $T$ is some nondecreasing function. Typically, we assume
that $T(m)=o(m)$, \emph{i.e.} $T$ is strictly better than a linear search.

We now consider the problem of finding the first marked element in $S$,
with probability at least, say, $1/2$. A trivial algorithm is to perform
a linear search in $O(m)$ until $\tilde{g}$ returns $1$. If $\tilde{g}$
had no errors, we could perform a binary search in $O(T(m))$.
This does not work very well in the presence of errors because decisions
made are irreversible, and errors accumulate quickly. Our observation
is that if we modify the binary search to boost the success probability of certain calls to $\tilde{g}$, we can still solve the problem in time in
$O\left( T(m)\right)$.

 \subsection{Algorithm}
 
The idea is inspired by \cite{abikkpssv2020} and very similar to \cite{FRDU94}
except that the calls to $\tilde{g}$ (the ``comparisons'') do not necessarily have
unit cost. The difference with \cite{FRDU94} is apparent in the complexity since
we can avoid the extra $\log$ factor when the cost of each comparison becomes
high enough.
For reasons that become clear in the proof, we need to boost some calls' success probability. We do so by repeating them several times and
taking the majority: by this we mean that we take the most common answer,
and return an error in the case of a tie.

\begin{algorithm}
    \caption{Binary search for a function with two-sided
    errors\label{alg:bin_search_two_sided}}
    \begin{algorithmic}
        \State $\ell\gets 1, r\gets m+1$\Comment{search interval}
        \State $d\gets 1$ \Comment{depth of the search}
        \While{$\ell<r$}
            \State $mid \gets \lfloor(\ell+r)/2\rfloor $
            \State $v_l\gets \tilde{g}(\ell,mid)$
                \Comment{repeat $d$ times and take the majority}
            \If{$v_l=0$}
                \State $\ell\gets mid+1$
            \Else
                \State $r\gets mid$
            \EndIf
            \State $d\gets d+1$
        \EndWhile
    \end{algorithmic}
\end{algorithm}

\begin{proposition}\label{lm:bin-search}
    Assume that $T$ satisfies
    $T(m/k)=O(T(m)/k^\alpha)$ for some $\alpha>0$ and every $m$ and $k$,
    then
    with probability more than $0.5$,
    \Cref{alg:bin_search_two_sided} returns the position of the first marked
    element, or $m+1$ if none exists.
    The running time is $O(T(m))$. 
\end{proposition}
\begin{proof}
The correctness of the algorithm, when there are no errors, is clear.
We need to argue about the complexity and error probability.

At the $u^{th}$ iteration of the loop, the algorithm considers a segment
$[\ell,r]$ of length at most $m\cdot 2^{-(u-1)}$.
The complexity of $\tilde{g}(\ell,mid)$ is at most $O(T(r-\ell))=O\left(T(m\cdot 2^{-(u-1)})\right)$
but we repeat it $2u$ times, so the total complexity of the $u^{th}$ iteration is $O\left(uT\left(m\cdot 2^{-(u-1)}\right)\right)$.
The number of iterations is at most $\log_2 m$. Hence, the total complexity
is
\begin{align*}
    O\left(\sum_{u=1}^{\log_2 m}uT\left(m\cdot 2^{-(u-1)}\right)\right)
    &=O\left(\sum_{u=1}^{\log_2 m}T\left(m\right)2^{-\alpha u}u\right)
    =O\left(T\left(m\right)\sum_{u=1}^{\infty}2^{-\alpha u}u\right)\\
    &=O\left(T(m)\frac{2^{\alpha}}{(2^{\alpha}-1)^2}\right)
    =O\left(T(m)\right).
\end{align*}

Finally, we need to analyze the success probability of the algorithm:
at the $u^{th}$ iteration, the algorithm will run each test $2u$ times
and each test has a constant probability of failure $p$.
Hence
for the algorithm to fail at iteration $u$, at least half of the
$2u$ runs must fail: this happens with probability at most
$
    {2u\choose u}p^u
    \leqslant \left(\frac{2ue}{u}\right)^up^u
    \leqslant (2ep)^u,
$
where $e=\exp(1)$.
Hence, the probability that the algorithm fails is bounded by
$
    \sum\limits_{u=1}^{\log_2 m}(2ep)^u
    \leqslant\sum_{u=1}^{\infty}(2ep)^u
    \leqslant \frac{2ep}{1-2ep}
$.
By taking $p$ small enough (say $2ep<\tfrac{1}{3}$), which
is always possible by repeating the calls to $\tilde{g}$ a constant number
of times to boost the probability, we can ensure that the algorithm fails less than half
of the time.
\end{proof} 

\begin{remark}
    The condition $T(m/k)=O(T(m)/k^\alpha)$ for some $\alpha>0$ and every $m$ and $k$ is clearly satisfied by any function of the form $T(m)=m^\alpha\log^\beta m\log^\gamma\log m$.
\end{remark}

\subsection{Application to Quantum Search}\label{sec:first-one}

A particularly useful application of the previous section is for quantum
search, particularly when $\tilde{g}$ is a Grover-like search. Indeed,
Grover's search can decide in $O(\sqrt{m})$ queries if a marked element exists
in an array of size $m$, with a constant probability of error. 

More precisely, assume that we have a function $f:\{1,\dots m\}\to\{0,1\}$
and the task is to find the minimal $x\in\{1,\dots, m\}$ such that $f(x)=1$.
If we let $\tilde{g}(\ell,r)=\textsc{GROVER}(\ell,r,f)$ then $\tilde{g}$
has query complexity $T(m)=\sqrt{m}$ and fails with constant probability.
Hence, we can apply \Cref{lm:bin-search} and obtain an algorithm to
find the first marked element with query complexity $T(m)=O(\sqrt{m})$ and constant probability of error. In fact, note that we are not making
use of \Cref{lm:bin-search} to its full strength because $\tilde{g}$ really
has one-sided error: it will never return $1$ if there are no marked
element. We will make use of this observation later. 

\begin{proposition}\label{lm:the-first-one-base}
    There is a quantum algorithm that finds the first marked element in an array of size $m$ in $O(\sqrt{m})$ queries and error probability
    $<0.5$.
\end{proposition}

As observed above, we are not really using \Cref{lm:bin-search} to its full
strength because Grover's search has one-sided error. This suggests that
there is room for improvement. Suppose that we now only have access to
a two-sided probabilistic approximation $\tilde{f}$ of $f$. In other words,
$f$ can now make mistakes: it can return $1$ for an unmarked element or
$0$ for a marked element with some small probability. Formally,
$\tilde{f}(x)=f(x)$ with probability at least $1-p$ and $1-f(x)$ otherwise,
for some probability $p<1/2$. We cannot apply Grover's search directly
in this case, but some variants have been
developed that can handle bounded errors \cite{PMR03}. Using this result,
we can build a two-sided error function $\tilde{g}$ with high probability
of success and time complexity $O(\sqrt{m})$. Applying \Cref{lm:bin-search} again,
we obtain the following improvement:

\begin{proposition}\label{prop:the-first-one-base-two-sided}
    There exists a quantum algorithm $\textsc{FindFirst}$
    that finds the first marked element
    in an array of size $m$ in $O(\sqrt{m})$ queries and error probability less than $0.5$;
    even when the oracle access to the array has a two-sided error.
\end{proposition}

In quantum computing, $f$ rarely has
two-sided errors. For instance, Grover's search has a one-sided error only.
If we assume that $\tilde{f}$ has one-sided error only, we can obtain
a slightly better version of \Cref{prop:the-first-one-base-two-sided}.
Formally, we assume that 
$\tilde{f}(x)=f(x)$ with probability at least $1-p$ and $0$ otherwise.
%

\begin{proposition}\label{prop:find-first}
    There exists a quantum algorithm that finds the first marked element
    in an array of size $m$ in expected $O(\sqrt{x})$ queries and with error probability less than $0.5$, where $x$ is the position of the first marked element, or
    $O(\sqrt{m})$ queries if none is marked. Furthermore, it works
    even when the oracle access to the array has one-sided error.
    Additionally, it has a worst-case query complexity of $O(\sqrt{m})$ in all cases. 
\end{proposition}
\begin{proof}
Let $\textsc{FindFirst}$ denote the algorithm from
\Cref{prop:the-first-one-base-two-sided} and $\textsc{GroverTwoSided}$
denote the variant of Grover's algorithm of \cite{PMR03} that works with
two-sided error oracles. Recall that we assume that $\tilde{f}$ has
one-sided error, \emph{i.e.} it may return $0$ instead of $1$ with
small probability but not the other way around. Consider the following
algorithm:

\begin{algorithm}[H]
    \caption{$\textsc{FindFirstAdvanced}(m,f)$. Find the first marked element in an array.\label{alg:first-one}}
    \begin{algorithmic}
        \State $r\gets 1$\Comment{size of the search space}
        \While{$r\leq n$ and $\textsc{GroverTwoSided}(1,r,\tilde{f})=0$}
            \State $r\gets \min(m,2r)$
        \EndWhile 
        \State \Return{$\textsc{FindFirst}(r,\tilde{f})$}
    \end{algorithmic}
\end{algorithm}

We now show that this algorithm satisfies the requirements of
\Cref{prop:find-first}. To simplify the proof, we assume that the array
always contains a marked element; this is without loss of generality
because we can add an extra object at the end that is always marked.
Furthermore, we assume that $n$ is a power 2, this is again without loss
of generality because we can add dummy object at the end at the cost of doubling
the array size at most.

Recall that $\tilde{f}$ has a one-sided error,
and the same applies to $\textsc{GroverTwoSided}$ in this case. Therefore
the test $\textsc{GroverTwoSided}(1,r,\tilde{f})=0$ can only fail if
there actually is a marked element in the interval $[1,r]$. Of course, the problem is that it can succeed even though there is a marked element
in this interval. Let $p$ be the probability that this happens (\emph{i.e.}
$\textsc{GroverTwoSided}$ fails), we know that this is $<1/2$ by
\cite[Theorem~10]{PMR03}. 
Let $x$ be the position of the first marked element and let $\ell_x$ be
such that $2^{\ell_x}\leqslant x<2^{\ell_x+1}$.
Let $R$ be the value of $r$ after the loop, it is a random variable and always a power of $2$.
By the above reasoning, it is always the case
that $R\geqslant x$. Furthermore, for any $\ell_x\leqslant\ell<\log_2n$,
the probability that $R=2^\ell$ is at most $p^{\ell-\ell_x}(1-p)$.
The call to $\textsc{FindFirst}$ takes time $O(\sqrt{R})$ by \Cref{prop:the-first-one-base-two-sided}. Hence, the expected time complexity
of this algorithm is
\begin{align*}
    O\left(\sum_{\ell=\ell_x}^{\log n}p^{\ell-\ell_x}(1-p)\sqrt{2^\ell}\right)
    &=O\left(\sqrt{2^{\ell_x}}\sum_{\ell=0}^{\infty}p^{\ell}\sqrt{2^\ell}\right)\\
    &=O\left(\sqrt{2^{\ell_x}}\frac{1}{1-\sqrt{2}p}\right)\\
    &=O\left(\sqrt{x}\right)
\end{align*}
where we assume that $p$ is small enough. This is always possible by repeating
the calls to $\textsc{FindFirst}$ a constant number of times to reduce the
failure probability $p$. Finally, we note that the only way this algorithm
can fail is if the (unique) call to $\textsc{FindFirst}$ fails and this only
happen with constant probability.
%
 
 
\end{proof}

\section{A Fast Quantum Implementation of  Online Algorithm for $k$-server Problem on Trees}\label{sec:q-algorithm}

We consider a special way of storing a rooted tree. Assume that for each vertex $v$ we have access to a sequence $a^v=(a^v_1,\dots,a^v_d)$ for $d=\dist(1,v)+1$. Here $a^v$ is a path from the root (the vertex $1$) to the vertex $v$, $a^v_1=1, a^v_d=v$. Such a way of describing a tree is not uncommon, for example 
when the tree represents a file system. A file path
``\texttt{c:\textbackslash Users\textbackslash MyUser\textbackslash newdoc.txt}'' is exactly
such a path in the file system tree.
Here ``\texttt{c}'', ``\texttt{Users}'', ``\texttt{MyUser}'' are ancestors of ``\texttt{newdoc.txt}'', ``\texttt{c}'' is the root and  ``\texttt{newdoc.txt}'' is the node
itself. 

We assume that we have access to the following two oracles in time $O(1)$:
(i) given a vertex $u$, a (classical) oracle that returns the length of the string path $a^u$; (ii) given a vertex $u$ and an index $i$, a quantum oracle that returns the $i^{th}$ vertex $a_i^u$ of the sequence $a^u$. 

We can solve the $k$-server problem on trees using the same algorithm
as in Section \ref{sec:c-algorithm} with the following modifications:
(i) The function $\LCA(u,v)$ becomes $\textsc{LCP}(a^u,a^v)$ where $\textsc{LCP}(a^u,a^v)$ is a longest common prefix of two sequences $a^u$ and $a^v$. (ii) $\textsc{MoveUp}(v,z)$ is the vertex $a^v_{d-z}$ where $a^v$ is the sequence for $v$ and $d=|a^v|$; (iii) We can compute $\dist(u,v)$ if $u$ is the ancestor of $v$: it is $d'-d''$, where $d'$ is a length of $a^u$ and $d''$ is a length of $a^v$. Note that the invocations of $\dist$ in Algorithms \ref{alg:move} and \ref{alg:query} are always this form. The only exception is $\dist(u,v)$ in $\textsc{Sort}$ in which the function uses $\textsc{LCA}$ as a subroutine. The complexity of $\textsc{Sort}$ is thus the same as the complexity of $\textsc{LCA}$ or $\textsc{LCP}$ in our case.

By doing so, we do not need any preprocessing. We now replace the $\textsc{LCP}(a^u,a^v)$ function by a quantum
subroutine $\textsc{QLCP}(a^u,a^v)$, presented in \Cref{sec:qlcp},
and keep everything else as is. This subroutine runs in time
$O(\sqrt{n}\log n)$ with $O\left(\frac{1}{n^3}\right)$ error probability.
This allows us to obtain the following result.

\begin{theorem}\label{th:quantum-compl}
    There is a quantum algorithm that processes a request in time
    $O\left(k^2\sqrt{n}\log n\right)$ with probability of
    error $O\left(\frac{1}{n}\right)$. It
    does not require any preprocessing.
\end{theorem}

\begin{proof}
 The complexity $\textsc{Move}$ is the complexity of $\textsc{LCA}$
 that is $\textsc{QLCP}$ in our implementation, plus the complexity of $\textsc{MoveUp}$. The former has complexity is $O(\sqrt{n}\log n)$ by Lemma \ref{lm:lcp}, and the latter
 $O(1)$ by the oracle. Therefore, the total running time of $\textsc{Move}$ is $O(\sqrt{n}\log n)$.
 
 The complexity of $\textsc{Request}$ is $O(k^2)$ times the cost of $\textsc{LCA}$ that is $\textsc{QLCP}$ in our implementation, and then a call to $\textsc{Move}$.  Additionally, the $\textsc{Sort}$ function invokes $\textsc{LCA}$ to compute distances. Hence, the complexity of $\textsc{Sort}$ is $O\left(k\log k\cdot \sqrt{n}\cdot\log n\right)$, and the total complexity is $O\left(k^2\sqrt{n}\cdot\log n\right)$.
 
 We invoke, $\textsc{QLCP}$ at most $4k^2$ times so the success probability is at least $\left(1-\frac{1}{n^3}\right)^{4k^2}\geq \left(1-\frac{1}{n^3}\right)^{4n^2}\geq 1-\Omega\left(\frac{1}{n}\right)$ for enough big $n$. Therefore, the error probability is $O\left(\frac{1}{n}\right)$.
 Note that we do not need any preprocessing.
\end{proof}

\subsection{Quantum Algorithm for Longest Common Prefix of Two Sequences}\label{sec:qlcp}

Let us consider the Longest Common Prefix (LCP) problem. Given two sequences $(q_1,\dots,q_d)$ and $(b_1,\dots,b_s)$, the problem is to find $t$ such that $q_1=b_1,\dots,q_t=b_t$ and $q_i\neq b_i$ for $t+1\leq i \leq m$, where $m=min(d,s)$.
Consider a function $f:\{1,\dots,m\}\to\{0,1\}$ such that $f(i)=1$ iff $q_i\neq b_i$. Assume that $x$ is the minimal argument such that $f(x)=1$, then $t=x-1$. The LCP problem is equivalent to finding the first marked element. It can be solved as follows.

\begin{algorithm}[h]
    \caption{$\textsc{QLCP}(q,b)$. Quantum algorithm for the longest common prefix.
    \label{alg:lcp}}
    \begin{algorithmic}
      \State $m\gets \min(d,s)$
      \State $x\gets \textsc{FindFirst}(m,f)$ \Comment{Repeat $3\log m$ times and take the majority vote}
      \If{$x=NULL$}
        \State $x\gets m+1$
      \EndIf
      \State \Return{$x-1$}
    \end{algorithmic}
\end{algorithm}

\begin{lemma}\label{lm:lcp}
Algorithm \ref{alg:lcp} finds the LCP of two sequences of length $m$ in time $O(\sqrt{m}\log m)$ and with probability of error $O\left(\frac{1}{m^3}\right)$.
\end{lemma}
\begin{proof}
The correctness of the algorithm follows from the definition of $f$. 
The complexity of $\textsc{FindFirst}$ is $O(\sqrt{m})$ by \Cref{lm:the-first-one-base}. The total running time is $O(\sqrt{m}\log m)$ because of the repetitions.
\end{proof}

%
%
%
\bibliographystyle{splncs04}
\bibliography{biblio}

\begin{thebibliography}{10}
\providecommand{\url}[1]{\texttt{#1}}
\providecommand{\urlprefix}{URL }
\providecommand{\doi}[1]{https://doi.org/#1}

\bibitem{aazksw2019part1}
Ablayev, F., Ablayev, M., Huang, J.Z., Khadiev, K., Salikhova, N., Wu, D.: On
  quantum methods for machine learning problems part i: Quantum tools. Big Data
  Mining and Analytics  \textbf{3}(1),  41--55 (2019)

\bibitem{aakv2018}
Ablayev, F., Ablayev, M., Khadiev, K., Vasiliev, A.: Classical and quantum
  computations with restricted memory. LNCS  \textbf{11011},  129--155 (2018)

\bibitem{a2017}
Ambainis, A.: Understanding quantum algorithms via query complexity. In: Proc.
  Int. Conf. of Math. 2018. vol.~4, pp. 3283--3304 (2018)

\bibitem{abikkpssv2020}
Ambainis, A., Balodis, K., Iraids, J., Khadiev, K., Klevickis, V., Prusis, K.,
  Shen, Y., Smotrovs, J., Vihrovs, J.: Quantum lower and upper bounds for
  2d-grid and {D}yck language. In: Proceedings of MFCS 2020. LIPIcs, vol.~170,
  pp. 8:1--8:14 (2020)

\bibitem{gs93}
Baliga, G.R., Shende, A.M.: On space bounded server algorithms. In: Proceedings
  of ICCI'93. pp. 77--81. IEEE (1993)

\bibitem{bk2009}
Becchetti, L., Koutsoupias, E.: Competitive analysis of aggregate max in
  windowed streaming. In: ICALP. LNCS, vol.~5555, pp. 156--170 (2009)

\bibitem{bfc2000}
Bender, M.A., Farach-Colton, M.: The lca problem revisited. In: Latin American
  Symposium on Theoretical Informatics. pp. 88--94. Springer (2000)

\bibitem{bf2004}
Bender, M.A., Farach-Colton, M.: The level ancestor problem simplified.
  Theoretical Computer Science  \textbf{321}(1),  5--12 (2004)

\bibitem{blm2015}
Boyar, J., Larsen, K.S., Maiti, A.: The frequent items problem in online
  streaming under various performance measures. International Journal of
  Foundations of Computer Science  \textbf{26}(4),  413--439 (2015)

\bibitem{bbht98}
Boyer, M., Brassard, G., H{\o}yer, P., Tapp, A.: Tight bounds on quantum
  searching. Fortschritte der Physik  \textbf{46}(4-5),  493--505 (1998)

\bibitem{cl91}
Chrobak, M., Larmore, L.L.: An optimal on-line algorithm for k servers on
  trees. SIAM Journal on Computing  \textbf{20}(1),  144--148 (1991)

\bibitem{cormen2001}
Cormen, T.H., Leiserson, C.E., Rivest, R.L., Stein, C.: Introduction to
  Algorithms. McGraw-Hill (2001)

\bibitem{dh96}
D{\"u}rr, C., H{\o}yer, P.: A quantum algorithm for finding the minimum.
  arXiv:quant-ph/9607014  (1996)

\bibitem{FRDU94}
Feige, U., Raghavan, P., Peleg, D., Upfal, E.: Computing with noisy
  information. SIAM Journal on Computing  \textbf{23}(5),  1001--1018 (1994)

\bibitem{fnn2006}
Flammini, M., Navarra, A., Nicosia, G.: Efficient offline algorithms for the
  bicriteria k-server problem and online applications. Journal of Discrete
  Algorithms  \textbf{4}(3),  414--432 (2006)

\bibitem{g96}
Grover, L.K.: A fast quantum mechanical algorithm for database search. In:
  Proceedings of STOC96. pp. 212--219. ACM (1996)

\bibitem{PMR03}
H{\o}yer, P., Mosca, M., de~Wolf, R.: Quantum search on bounded-error inputs.
  In: Automata, Languages and Programming (2003)

\bibitem{h95}
Hughes, S.: A new bound for space bounded server algorithms. In: Proceedings of
  the 33rd annual on Southeast regional conference. pp. 165--169 (1995)

\bibitem{kmrs86}
Karlin, A.R., Manasse, M.S., Rudolph, L., Sleator, D.D.: Competitive snoopy
  caching. In: FOCS, 1986., 27th Annual Symposium on. pp. 244--254. IEEE (1986)

\bibitem{k2021}
Khadiev, K.: Quantum request-answer game with buffer model for online
  algorithms. application for the most frequent keyword problem. In: CEUR
  Workshop Proceedings. vol.~2850, pp. 16--27 (2021)

\bibitem{kk2019}
Khadiev, K., Khadieva, A.: Two-way quantum and classical machines with small
  memory for online minimization problems. In: International Conference on
  Micro- and Nano-Electronics 2018. Proc. SPIE, vol. 11022, p. 110222T (2019)

\bibitem{kk2019disj}
Khadiev, K., Khadieva, A.: Quantum online streaming algorithms with logarithmic
  memory. International Journal of Theoretical Physics  \textbf{60},  608--616
  (2021)

\bibitem{kkm2018}
Khadiev, K., Khadieva, A., Mannapov, I.: Quantum online algorithms with respect
  to space and advice complexity. Lobachevskii J. of Math.  \textbf{39}(9),
  1210--1220 (2018)

\bibitem{kk2019online2w}
Khadiev, K., Khadieva, A.: Two-way quantum and classical automata with advice
  for online minimization problems. In: Reversibility in Programming,
  Languages, and Automata (RPLA 2019) procceedings. LNCS, Springer (2020)

\bibitem{kl2020}
Khadiev, K., Lin, D.: Quantum online algorithms for a model of the
  request-answer game with a buffer. Uchenye Zapiski Kazanskogo Universiteta.
  Seriya Fiziko-Matematicheskie Nauki  \textbf{162}(3),  367--382 (2020)

\bibitem{k2016}
Komm, D.: An Introduction to Online Computation: Determinism, Randomization,
  Advice. Springer (2016)

\bibitem{k2014}
Kothari, R.: An optimal quantum algorithm for the oracle identification
  problem. In: Proceedings of STACS 2014. p.~482 (2014)

\bibitem{ll2015}
Lin, C.Y.Y., Lin, H.H.: Upper bounds on quantum query complexity inspired by
  the elitzur-vaidman bomb tester. In: Proceedings of CCC 2015 (2015)

\bibitem{nc2010}
Nielsen, M.A., Chuang, I.L.: Quantum computation and quantum information.
  Cambridge univ. press (2010)

\bibitem{rbm2013}
Rudec, T., Baumgartner, A., Manger, R.: A fast work function algorithm for
  solving the k-server problem. Cent Eur J Oper Res  \textbf{21}(1),  187--205
  (2013)

\end{thebibliography}
\newpage
\appendix

\section{Implementation of Binary Lifting}\label{apx:bl}
 The $\textsc{BL\_Preprocessing}()$ prepares an array $up$ that stores data for $\textsc{MoveUp}$ subroutine. For a vertex $v$ and an integer $0\leq w \leq \lfloor\log_2 n\rfloor$, the cell $up[v][w]$ stores a vertex $v'$ on the path from $v$ to the root and at distance $\dist(v,v')=2^w$.
We construct the array using  dynamic programming and obtain the following formulas:
\[up[v][w]\gets up[up[v][w-1]][w-1],\quad up[v][0] \gets \textsc{Parent}(v)\]

Let us show that the formulas are correct. Let $v'=up[v][w]$, $v''=up[v][w-1]$. Then $\dist(v',v)=\dist(v'',v)+\dist(v'',v')=2^{w-1}+2^{w-1}=2^w$.

The algorithm is presented in Algorithm \ref{alg:bl-preproc}

\begin{algorithm}[H]
    \caption{$\textsc{BL\_Preprocessing}()$ prepares the required data structures for Binary Lifting technique.\label{alg:bl-preproc}}
    \begin{algorithmic}
        \For{$v\in V$}
        \State $up[v][0]\gets\textsc{Parent}(v)$
        \EndFor
        \For{$w\in\{1,\dots \lfloor \log_2 n\rfloor\}$}
        \For{$v\in V$}
        \State $v''\gets up[v][w-1]$
        \If{$v''= NULL$}
        \State $up[v][w]\gets NULL$
        \EndIf
        \If{$v''\neq NULL$}
        \State $up[v][w]\gets up[v''][w-1]$
        \EndIf
        \EndFor
         \EndFor
    \end{algorithmic}
\end{algorithm}

The subroutine $\textsc{MoveUp}(v,z)$ returns a vertex $v'$ on the path from $v$ to the root and at distance $dist(v',v)=z$. First, we find the maximal $w'$ such that $2^{w'}<z$. Then, we move to the vertex $up[v][w']$ and reduce $z$ by $2^{w'}$. We repeat this action until $z= 0$. The total number of steps is at most $O(w')=O(\log n)$.

\begin{algorithm}[H]
    \caption{$\textsc{MoveUp}(v,z)$ returns the vertex $v'$ at distance $z$ on the path to the root.\label{alg:bl-moveup}}
    \begin{algorithmic}
        \State $w\gets 0$
        \State $power2\gets 1$
        \While{$power2\cdot 2 \leq z$}
        \State $w\gets w+1$
        \State $power2\gets power2 \cdot 2$
        \EndWhile
       \While{$z\neq 0$}
       \State $v''\gets v$
             \While{$power2>z$}
        \State $w\gets w-1$
        \State $power2\gets power2 / 2$
        \EndWhile
            \State $v\gets up[v][w]$
            \State $z\gets z-dist(v'',v)$
       \EndWhile
       \State \Return $v$
    \end{algorithmic}
\end{algorithm}

\section{Implementation of $\textsc{ComputeDistance}(u)$ Subroutine for Distance Computing}\label{apx:alg:dist}
\begin{algorithm}[ht]
    \caption{$\textsc{ComputeDistance}(u)$. Recursively compute the distance from the root.} \label{alg:dist}
    \begin{algorithmic}
        \For{$v \in \textsc{Children}(u)$}
        \State $\dist(1,v)\gets \dist(1,u) + 1$
        \State $\textsc{ComputeDistance}(v)$ 
        \EndFor
    \end{algorithmic}
\end{algorithm}

\end{document}